\documentclass[10pt, journal]{IEEEtran}

\ifCLASSINFOpdf
 
\else
 
\fi

\usepackage{multicol}
\usepackage{lipsum}
\usepackage{setspace}
\usepackage{amsmath}
\usepackage{cases}
\usepackage{cite}
\usepackage{amsfonts}
\usepackage{amssymb}
\usepackage[linesnumbered,ruled,vlined]{algorithm2e}

\usepackage{array}
\usepackage[margin=1.0in]{geometry}
\usepackage{amsthm}
\usepackage{empheq}
\renewcommand*{\stackrel}{%
\mathrel\bgroup\stack@relbin
}

\newtheorem{theorem}{Theorem}[section]

\newtheorem{proposition}[theorem]{Proposition}
\newtheorem{corollary}[theorem]{Corollary}

\ifCLASSINFOpdf
   \usepackage[pdftex]{graphicx}

  \graphicspath{{../pdf/}{../jpeg/}{../eps/}}
  \DeclareGraphicsExtensions{.eps,.pdf,.jpeg,.png}
\else
   \usepackage[dvips]{graphicx}
   \graphicspath{{../eps/}}
   \DeclareGraphicsExtensions{.eps}
\fi
\usepackage{fmtcount}
\usepackage{subfigure}
\usepackage{array}
\usepackage{url}
\usepackage[linesnumbered,ruled,vlined]{algorithm2e}
\usepackage{scalerel}
\usepackage{multicol}
\usepackage{float}
\usepackage{lipsum}
\usepackage{mathtools}  
\usepackage{tabulary}
\usepackage{booktabs}
\usepackage{stfloats}
\usepackage{xcolor}
\usepackage{epstopdf}
\usepackage{epsfig}

\usepackage{graphicx}
\def\smallX{\fontsize{8.3pt}{8.3pt} \selectfont}

\begin{document}

\title{{\fontsize{17}{17}\selectfont Power Control and Soft Topology Adaptations in Multihop Cellular Networks with {Multi-Point Connectivity}}}

\author{Syed~Amaar~Ahmad and~Luiz~A.~DaSilva,~\IEEEmembership{Senior Member,~IEEE}% <-this % stops a space
\thanks{Syed A. Ahmad is currently working with Datawiz Corporation. Luiz A. DaSilva is with the Department of Electrical and Computer Engineering, Virginia Tech and also with CTVR, Trinity College, Dublin, Ireland. Email: $\{$saahmad, ldasilva$\}$@vt.edu.} 
\thanks{Manuscript received xxxx; revised xxxx}}

\markboth{IEEE Transactions on Communications,~Vol.~, No.~, ~}%
{Shell \MakeLowercase{\textit{et al.}}: Bare Demo of IEEEtran.cls for Journals}

\maketitle

\begin{abstract}
The LTE standards account for the use of relays to enhance coverage near the cell edge. In a traditional topology, a mobile can either establish a direct link to the base station (BS) or a link to the relay, but not both. In this paper, we consider the benefit of {\emph{multipoint connectivity}} in allowing User Equipment (UEs) to split their transmit power over simultaneous links to the BS and the relay, in effect transmitting two parallel flows. We model decisions by the UEs as to: (i) which point of access to attach to (either a relay or a relay and the BS or only the BS); and (ii) how to allocate transmit power over these links so as to maximize their total rate. We show that this  flexibility in the selection of points of access leads to substantial network capacity increase against when nodes operate in a fixed network topology. Individual adaptations by UEs, in terms of both point of access and transmit power, are interdependent due to interference and to the possibility of over-loading of the backhaul links. We show that these decisions can converge without any explicit cooperation and derive a closed-form expression for the transmit power levels.
\end{abstract}


\IEEEpeerreviewmaketitle

\section{Introduction}
\IEEEPARstart{L}{TE} networks anticipate the use of low-cost relays to increase the coverage region of a base station (e{N}ode{B}). In the currently deployed LTE architecture, each User Equipment (UE) can only connect to a single access point at a given time, despite the availability in the cellular system of multiple access points such as relays and e{N}ode{B}s \cite{handoverLTE}. {In this paper, we consider a flexible two-hop uplink network where UEs have \emph{multi-point connectivity} and choose between transmitting their data via either a single or multiple access points simultaneously, with transmit power adaptations.} In a multihop cellular network, the flexibility to dynamically choose between different access points enables nodes to better overcome the resource limitations introduced by the wireless channel and the relay backhaul link. 

If not adequately balanced, the link of a relay to the eNodeB (wireless backhaul) may limit the end-to-end data rate of UEs connected to the relay, especially if there is a large number of users. To support a higher load via the relay, a network operator may allocate a larger bandwidth to the backhaul to increase the uplink's capacity. However, in LTE-Advanced, the channels used for the relay backhaul and for the UEs come from the same pool \cite{3gpp1s}\cite{LTEOpPers2012}. Hence, such a re-allocation to the relay-to-eNodeB link would also reduce the bandwidth available to the UEs and thus this approach may not lead to an increase in the network-wide performance. Therefore, in lieu of increasing the bandwidth of the relay backhaul, we instead propose to allow a UE to maintain parallel connections to a relay and to the eNodeB, potentially transmitting two independent data streams on two separate channels.
\subsection{Contributions}
We devise a link-adaptive scheme where UEs are capable of having parallel links to the base station and to the relays. Depending on the load on the relay backhaul, the UE can choose to split its data into two streams, where one is sent directly to the base station and the other one to the relay, or it can send both data streams to a single access point (i.e. either a relay or the base station). Under this scheme, each UE attempts to maximize its own  achievable data rate by appropriately selecting its access points (topology adaptation) and by performing distributed power allocation. The power allocation essentially works in two modes: (1) waterfilling or (2) the UE iteratively re-allocates more power on its link to the base station whenever there is over-loading of the backhaul between the relay and the base station. We show significant performance improvement owing to this flexibility to adapt the topology and use power allocation in response to the load on the backhaul links.
\subsection{Background}
The approach taken in our work overlaps with aspects of cellular topology control \cite{Amzallag2013}\cite{Lorenzo2009} and {parallel relay channels} \cite{ScheinGallager2000_PRC}. In a conventional cellular topology, nodes maintain a connection to a single access point, which is often the one with the highest received SNR, and aim to either maintain a target SINR with minimum power \cite{bambos2000}\cite{foschini_93} or maximize some related utility function \cite{xiao_03}. Moreover, several kinds of relaying techniques have been studied for cellular networks to improve coverage and provide better end-to-end performance \cite{Peters2013_relays}. Our work incorporates elements of {parallel relay channels} which have been studied in the literature in the context of: (i) a single source-destination node pair with transmit power allocation over the sub-channels without interference from other source nodes \cite{Rezaei2010_PRC_sched, Bakanoglu2011_PRC_mul_rel, LiangPoor2007_resAlloc_maxmin, Liang2005_resAlloc, ChenYener2008_PRC_pow}, (ii) a Gaussian interference channel \cite{EtkinTse2008_GIC} for the two-user case (i.e. two source-destination pairs) where the capacity region is explored with a relay backhaul operating on an {out-of-band} channel from the users \cite{sahin2011_icobr, yeTian2012_icobr, Sahin2011_icobr2, Razaghi2013_ICOBR} and (iii) a general multiple access channel from multiple source nodes to a single shared relay and common destination \cite{SankarPoor2011_OppScd}. 

Our approach differs from each of the above categories in the following ways. Unlike in \cite{Amzallag2013}\cite{Lorenzo2009}, nodes maintain parallel links to multiple access points. Also unlike \cite{Rezaei2010_PRC_sched, Bakanoglu2011_PRC_mul_rel, LiangPoor2007_resAlloc_maxmin, Liang2005_resAlloc, ChenYener2008_PRC_pow}, we assume there can be an arbitrary number of source nodes (i.e. UEs) in the network which may create mutual co-channel interference. The capacity bounds on interference channels with relaying are provided by \cite{sahin2011_icobr, yeTian2012_icobr, Sahin2011_icobr2, Razaghi2013_ICOBR} for certain cases (e.g. 2-user). %However, the exact capacity for the general $n$-user case remains an open problem. 
In contrast, the motivation for our work is how UEs can autonomously make power allocations and select access points by taking into account network-wide conditions to maximize achievable data rate while treating interference as noise. 

In our work, UE transmissions to the same access point (e.g. base station or relay) are orthogonal in frequency, as we do not assume any underlying multi-user detection (MUD) or interference cancellation capability at the receivers. Moreover, in contrast to \cite{SankarPoor2011_OppScd,sahin2011_icobr, yeTian2012_icobr, Sahin2011_icobr2, Razaghi2013_ICOBR}, we study a generalized network where there can be any number of noisy and interference channels of unequal bandwidth and where the number of relays may be more than one. {In LTE, a UE may transmit to two access points through either a \emph{coordinated multipoint transmission} approach \cite{Lee_ComP_Nov12}\cite{3gppComP12} or a \emph{dual connectivity} approach \cite{LeanCarrier2013}. In the former, a given bandwidth resource is shared for  transmission to the different access points, whereas the latter requires two different bandwidth resources. Our \emph{multipoint connectivity} approach is a hybrid of the two: a UE has two bandwidth resources and it can transmit data to two access points.} 
\begin{figure}[!t]
        \centering\includegraphics[width=0.47\textwidth]{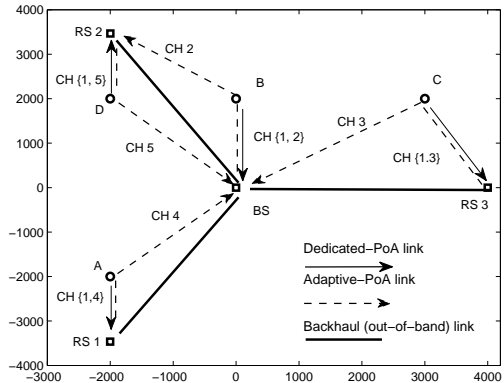}
\caption{A depiction of a network with parallel links. While all the UEs A, B, C and D transmit data on Channel 1 to dedicated-PoAs, they can send data to either the BS or to the RS on their other channel.}
\label{fig:furt1}
\end{figure}
In Section II, we present the system model, followed by the problem formulation in Section III. In Section IV, we propose the adaptation scheme for nodes to select their access points and allocate transmit power accordingly. In Section V, we show convergence of the link adaptations under certain constraints and derive a closed-form expression for the converged transmit power. In Section VI, we provide preliminary analysis of how our scheme would compare against cooperative MIMO, interference cancellation and centralized cross-layer optimization. In Section VII, we provide simulation results for our scheme. Finally, Section VIII provides the conclusions of the paper.

\section{System Model}
We consider a single-cell {uplink} network with a set ${\mathcal{N}}=\{1,2,\cdots, n\}$ of user equipment (UEs) that want to send their data to a base station (BS). Located towards the edges of the cell are $N_r$ fixed relays (RS) that provide coverage for UEs located farther away from the base station (see Fig. 1 with $N_r=3$). {All nodes are equipped with a single antenna}. All Points of Access (PoA) are denoted by a set ${\mathcal{R}}=\{1,2,\cdots,N_r, N_r+1\}$, where elements $1$ to $N_r$ indicate an RS and the BS is represented by element $N_r+1$. The RS decodes the data it receives from the UE and forwards it to the BS (i.e. a decode-and-forward scheme). 

Each UE can have at most two parallel links. A {node $i$ maintains its connection to a PoA $r_i \in {\mathcal{R}}$ on its first link, which is referred to as the \emph{dedicated-PoA} link. The PoA chosen by each UE on its dedicated-PoA link is the one with the highest received pilot signal strength on a downlink control channel. On its second link (referred to as the \emph{adaptive-PoA} link), node $i$ may choose to send data to the same PoA or to an alternative PoA. The UE can transmit either to the BS or to an RS on its adaptive-PoA link (i.e., $r_i^A \in {\mathcal{R}}$).

All links operate on a set of channels ${\mathcal{F}}=\{1,2,\cdots,F\}$. The assignment of each channel $f$ to nodes is represented as a set ${\mathcal{A}}(f) \subseteq {\mathcal{N}}$. To better manage intra-cell interference, some channels are re-used whereas some are exclusively assigned to UEs. We assume that the dedicated-PoA link of a node may experience intra-cell interference, due to frequency re-use, but its adaptive-PoA link operates on a channel which is orthogonal with respect to all other links, so that it is only noise-limited. 

All links in the network (i.e., UE-to-RS and UE-to-BS) can adapt their transmit power and, consequently, their data rate. Given the UEs' available power vector as ${\bf{P}}_{\max}$ Watts, the transmit power of nodes on their first link are represented by the vector ${\bf{P}}_d=\left[P^{(d)}_{1}, P^{(d)}_{2}, \cdots, P^{(d)}_{n}\right]$ and those on their adaptive-PoA links as ${\bf{P}}_a=\left[P^{(a)}_{1}, P^{(a)}_{2}, \cdots, P^{(a)}_{n}\right]$ where $P^{(d)}_{i}+P^{(a)}_{i}\leq P_{\max, i}$.

The channel gain between the UE node $j\in {\mathcal{A}}(f)$ and the intended receiver (PoA) of node $i \in {\mathcal{A}}(f)$ is represented as $h^{(f)}_{j, r_i}$, which depends on several factors such as shadowing, path loss and fading on channel $f\in {\mathcal{F}}$. For the adaptive-PoA link of UE $i$, we denote as $h^{(q)}_{i, r}$ and $h^{(q)}_{i, b}$ its gain to the RS and to the BS,  respectively, on channel $q \in \mathcal{F}$. As an example, consider Fig. \ref{fig:furt1}, where each node has a dedicated-PoA link on channel $1$.  {For UE $i$, given ${{x}}^{(d)}_i$ and ${{x}}^{(a)}_i$ as unit-power complex-valued input symbols sent on dedicated-PoA and adaptive-PoA links respectively, the output symbols ${{y}}^{(d)}_i$ and ${{y}}^{(a)}_i$ are:
\begin{equation} \label{effintf}
\small
\begin{split}
{{y}}^{(d)}_i &= h^{(f)}_{j, r_i}\sqrt{P^{(d)}_i}x^{(d)}_i + \sum_{\forall j\neq i} \sqrt{P^{(d)}_j}h^{(f)}_{j, r_i}x^{(d)}_j+\nu^{(d)}_i \\
&\mbox{, dedicated-PoA link}\\
{{y}}^{(a)}_i &= \sqrt{P^{(a)}_i}h^{(q)}_{i, r}x^{(a)}_i +\nu^{(r)}_i \mbox{, if adaptive-PoA link to RS}\\
{{y}}^{(a)}_i &= \sqrt{P^{(a)}_i}h^{(q)}_{i, b}x^{(a)}_i +\nu^{(b)}_i \mbox{, if adaptive-PoA link to BS}
\end{split}
\end{equation}
where $\nu^{(d)}_{i}$, $\nu^{(b)}_i$ and $\nu^{(r)}_i$ denote the zero-mean complex Gaussian noise on the dedicated-PoA and the adaptive-PoA links respectively.}

Corresponding to each channel gain, we define the power gain as $g^{(f)}_{j, r}= |h^{(f)}_{j, r}|^2$. For the dedicated-PoA links, let ${\bf{F}}$ be a normalized cross-link gain matrix of dimensions $n \times n$ 
\begin{equation} \label{gainmatrix1}
{\bf{F}}(i,j) =
\begin{cases}
0 \mbox{    if $i= j$ or $i,j$ orthogonal} \\
\frac{g^{(f)}_{j, r_i}}{g^{(f)}_{i, r_i}} \mbox{  otherwise},
\end{cases}
\end{equation}
where $i,j$ are orthogonal if their links operate on different channels in ${\mathcal{F}}$.} We also define ${\bf{D}}_d$, ${\bf{D}}_{r}$ and ${\bf{D}}_b$ as $n \times 1$ vectors, which represent normalized noise powers on the dedicated-PoA links and the adaptive-PoA links (to RS and BS) respectively. {The $i^{th}$ elements of these vectors are such that ${\bf{D}}_d(i)= n_{d}/g^{(f)}_{i, r_i}$, ${\bf{D}}_{r}(i)= n_{a}/g^{(q)}_{i, r}$ and ${\bf{D}}_b(i)= n_{a}/g^{(q)}_{i, b}$ where $n_{d} =  \mathrm{Var}[{\nu^{d}}_i]$ and $n_{a}= \mathrm{Var}[{\nu^{r}}_i]=\mathrm{Var}[{\nu^{b}}_i]$ are the thermal noise powers respectively}. 
\newcounter{tempequationcounter}
\begin{figure*}[!b]
\hrulefill
\small
\setcounter{tempequationcounter}{\value{equation}}
\begin{eqnarray}
\setcounter{equation}{5}
\begin{split}
\eta_N = \sum^{N_r}_{r=1} \min\left(\eta_r, \sum_{\forall i:r_i=r}W_d\eta\left(\gamma^{(d)}_{i}\right)+ \sum_{\forall i: r^A_i=r}\left(1-a_i\right)W_a\eta\left(\gamma^{(r)}_{i}\right)\right) +\sum_{\forall i} a_i W_a\eta\left(\gamma^{(b)}_{i}\right)+\sum_{\forall i: r_i=N_r+1} W_d\eta\left(\gamma^{(d)}_{i}\right)
\end{split}\label{maxflow}
\end{eqnarray}
\setcounter{equation}{\value{tempequationcounter}}
%\vspace*{4pt}
\end{figure*}
\newcounter{tempequationcounter2}
\begin{figure*}[!b]
\small
\setcounter{tempequationcounter2}{\value{equation}}
\setcounter{equation}{7}
\begin{equation}
\begin{split}\label{opt_func1}
f_i(a_i, P^{(d)}_{i}, P^{(a)}_{i})=\hspace*{7cm}\\
\begin{cases}
\left(1-a_i\right)\min\left((V_{r_i}+q_i)^+, W_d\eta\left(\gamma^{(d)}_{i}\right)+W_a\eta\left(\gamma^{(r)}_{i}\right) \right)+a_i \left(W_a\eta\left(\gamma^{(b)}_{i}\right)+\min \left((V_{r_i}+q_i)^+, W_d\eta\left(\gamma^{(d)}_{i}\right) \right)\right)\mbox{, $r_i\leq N_r$}\\
\left(1-a_i\right)\min\left((V_{r^A_i}+q_i)^+, W_a\eta\left(\gamma^{(r)}_{i}\right)\right) + a_iW_a\eta\left(\gamma^{(b)}_{i}\right)+W_d\eta\left(\gamma^{(d)}_{i}\right)\mbox{, $r_i= N_r+1$}
\end{cases}
\end{split}
\end{equation}
\setcounter{equation}{\value{tempequationcounter2}}
\hrulefill
%\vspace*{4pt}
\end{figure*}
\newcounter{tempequationcounter3}
The effective interference for node $i$ on its dedicated-PoA link and on its adaptive-PoA link is then defined as \cite{sung_05}:
\begin{equation} \label{effintf}
\small
E^{(d)}_{i} = \frac{n_d + \sum_{j\neq i, j \in {\mathcal{A}}(f) }g^{(f)}_{j,r_i}P^{(d)}_{j}}{g^{(f)}_{i, r_i}},
\end{equation}
\begin{displaymath}
\small
E^{(b)}_{i} = \frac{n_a}{g^{(q)}_{i, b}} 
\end{displaymath}
\begin{displaymath}
\small
E^{(r)}_{i} = \frac{n_a}{g^{(q)}_{i, r}}.
\end{displaymath}
With transmitter-side channel state information (CSIT), each UE knows these effective interferences. The corresponding SINRs are:
\begin{eqnarray}\label{sinr}
\small
\begin{split}
\gamma^{(d)}_{i} = \frac{P^{(d)}_{i}}{E^{(d)}_{i}}=\frac{P^{(d)}_i}{\frac{n_d + \sum_{j\neq i, j \in {\mathcal{A}}(f)}g^{(f)}_{j, r_i}P^{(d)}_i}{g^{(f)}_{i, r_i}}}\\ %\mbox{ (UE-to-RS link)}\\
\end{split}
\end{eqnarray}
\begin{displaymath}
\small
\gamma^{(b)}_{i} = \frac{P^{(a)}_i}{E^{(b)}_{i}}=\frac{P^{(a)}_i}{\frac{n_a}{g^{(q)}_{i, b}}} \\%\mbox{ (UE-to-BS link)}
\end{displaymath}
\begin{displaymath}
\small
\gamma^{(r)}_{i} = \frac{P^{(a)}_i}{E^{(r)}_{i}}=\frac{P^{(a)}_i}{\frac{n_a}{g^{(q)}_{i, r}}}
\end{displaymath}
where either $\gamma^{(b)}_{i}$ or $\gamma^{(r)}_{i}$ will be the achieved SINR depending on the choice of PoA by the UE on its adaptive-PoA link, and $\gamma^{(d)}_i$ will be the SINR on the dedicated-PoA link. Note that for a single-cell system, the dedicated-PoA links are interference-limited while the adaptive-PoA links are noise-limited. 

{The total spectrum resource for the UEs in ${\mathcal{F}}$ is fixed at $W_s$ Hz. From this pool, the channel bandwidth $W_d$ Hz of each dedicated-PoA link and that of each adaptive-PoA link $W_a$ Hz depend on the total number of UEs; when $n$ is large, the channel bandwidths $W_d$ and $W_a$ are smaller.} We define the achievable rate on either link as $W\cdot\eta(\gamma)=W\log_2(1+\gamma)$ where $W\in\{W_d, W_a\}$. % is the channel bandwidth of the link. 

Finally, the RS-to-BS links operate on \emph{out-of-band} channels from the UEs, and are assumed to have a capacity $\eta_r$ bps.  
\section{Problem Formulation}
If UE $i$ has a dedicated-PoA link to an RS then it has the choice of (a) either transmitting both its data streams to the same RS or (b) sending one stream to the RS and the other one to the BS. Conversely, if UE $i$ has a dedicated-PoA link to the BS, then it can also send one stream to an RS and the other one to the BS or (c) it can transmit both its data streams to the BS. Henceforth, we denote (a),(b) and (c) as RS-RS, BS-RS and BS-BS transmissions, respectively. 

\begin{table}[t]\caption{Symbol definitions: Subscript $i$ denotes UE $i$.}
\begin{center}% used the environment to augment the vertical space
% between the caption and the table
\begin{tabular}{r c p{0.01cm} }
\toprule
$a_{i}$& adaptive-PoA link to the BS\\
$P^{(d)}_{i}$& dedicated-PoA link transmit power\\
$P^{(a)}_{i}$& adaptive-PoA transmit power\\
$E^{(d)}_{i}$& effective interference dedicated-PoA link\\
$E^{(b)}_{i}$& effective interference adaptive-PoA (BS)\\
$E^{(r)}_{i}$& effective interference adaptive-PoA (RS)\\
$\gamma^{(d)}_{i}$& SINR dedicated-PoA link\\
$\gamma^{(b)}_{i}$& SNR adaptive-PoA link to BS\\
$\gamma^{(r)}_{i}$& SNR adaptive-PoA link to RS\\
$\tau^{(br)}_{i}$& peak data rate BS-RS transmission\\
$\tau^{(rr)}_{i}$& peak data rate RS-RS transmission\\
$\tau^{(bb)}_{i}$& peak data rate BS-BS transmission\\
$V_r$&  rate differential at RS $r$\\
$\widehat{V}$ & rate differential threshold\\
$r_i$& dedicated-PoA link receiver\\
$\eta(\gamma)$& $\log_2(1+\gamma)$\\
$\eta_r$& relay backhaul capacity\\
$q_i$& Data rate sent via relay\\
$W_d$ & Channel bandwidth dedicated-PoA link\\
$W_a$ & Channel bandwidth adaptive-PoA link\\
\bottomrule
\end{tabular}
\end{center}
\label{tab:TableOfNotationForMyResearch}
\end{table}
We next define a Boolean variable $a_i \in \{0,1\}$ for a UE $i$. If UE $i$ transmits to the BS on the adaptive-PoA link, then $a_i=1$; otherwise, if the PoA is a relay, $a_i=0$. Using the Max-Flow Min-Cut Theorem \cite{sherali}, the aggregate end-to-end data rate (network capacity), denoted as $\eta_N$, is formally defined in equation (\ref{maxflow})\addtocounter{equation}{1}. This expression represents the achievable sum rate in the network when interference is treated as noise. Next, we define the \emph{rate differential} across the RS $r \in \{1,2,\cdots N_r\}$ as 
\begin{equation}\label{Vr}
{\smallX
\begin{split}
V_r = \eta_r-\sum_{\forall i:r_i=r} W_d\eta\left(\gamma^{(d)}_{i}\right)- \sum_{\forall i:r^A_i=r} \left(1-a_i\right)W_a\eta\left(\gamma^{(r)}_{i}\right)
\end{split}}
\end{equation}
which denotes the difference between the outgoing RS-to-BS link capacity and the aggregate incoming data rate of associated UEs at relay $r$. When $V_r < 0$, the RS-to-BS link represents a \emph{bottleneck} link that limits the end-to-end data rate for the UEs sending data to the BS via RS $r$. Conversely, $V_r \geq 0$ indicates that the capacity of the RS-to-BS link is high enough to support the data rate load on the relay. Note that the UE-to-BS links (i.e. when nodes transmit directly to the BS) are also deemed to be bottleneck links. The aggregate end-to-end data rate $\eta_N$ in equation (\ref{maxflow}) is the sum of the data rates of all bottleneck links in the network.

When UEs adapt transmit power or change PoAs, the data rates at which they can transmit to their RS will vary, thereby causing a change in the rate differential. When $V_r \geq 0$, there is room for UE $i$ to send at a higher data rate to its RS. Conversely, when $V_r < 0$, the UE cannot increase its end-to-end rate by forwarding its data exclusively through the relay. In this case, it may  choose to transmit some of its data to the BS directly and allocate its transmit power accordingly. %The objective of each UE $i$ is to maximize its own end-to-end rate in  (\ref{opt_func1}). %as shown in (\ref{opt_func2}). 
We denote the data rate sent to the relay associated with a UE $i$ as 
\setcounter{equation}{6}
\begin{equation}
q_i = 
\begin{cases}
W_d\eta(\gamma^{(d)}_i)+(1-a_i)W_a\eta(\gamma^{(r)}_i) \mbox{, $r_i\leq N_r$}\\
a_iW_a\eta(\gamma^{(r)}_i) \mbox{,  $r_i= N_r+1$.}
\end{cases}
\end{equation}
By adjusting the rate differential for $q_i$, we denote $(V_r+q_i)^+$ as the available backhaul capacity at relay $r$ that UE $i$ can achieve given the current load caused by other nodes. We therefore set the adaptation objective of each node $i$ as given by (\ref{opt_func1})\addtocounter{equation}{1}, which reduces in the following manner:
\newcounter{tempequationcounterPeakDATA}
\begin{figure*}[!b]
\normalsize
\hrulefill
\setcounter{tempequationcounterPeakDATA}{\value{equation}}
\setcounter{equation}{9}
\begin{equation}{\small
\begin{split}\label{peakData}
\tau^{(rr)}_i=\underbrace{\max}_{P^{(d)}_{i}, P^{(a)}_{i}} W_d\log_2\left(1+\frac{P^{(d)}_{i}}{E^{(d)}_{i}}\right)+W_a\log_2\left(1+\frac{P^{(a)}_{i}}{E^{(r)}_{i}}\right) \mbox{, RS-RS transmission $r_i\leq N_r$} \hspace*{6.3cm} \\
\tau^{(bb)}_i=\underbrace{\max}_{P^{(d)}_{i}, P^{(a)}_{i}} W_d\log_2\left(1+\frac{P^{(d)}_{i}}{E^{(d)}_{i}}\right)+W_a\log_2\left(1+\frac{P^{(a)}_{i}}{E^{(b)}_{i}}\right)\mbox{, BS-BS transmission $r_i= N_r+1$}   \hspace*{5.7cm} \\
\tau^{(br)}_i=\underbrace{\max}_{P^{(d)}_{i}, P^{(a)}_{i}}  W_d\log_2\left(1+\frac{P^{(d)}_{i}}{E^{(d)}_{i}}\right)+W_a\log_2\left(1+\frac{P^{(a)}_{i}}{E^{(x)}_{i}}\right) \mbox{, BS-RS transmission $x\in \{b, r\}$}\hspace*{6cm} \\
\mbox{$x=b$ if $r_i\leq N_r$, $x=r$ if $r_i=N_r+1$} \hspace*{10cm}
\end{split}}
\end{equation}
\setcounter{equation}{\value{tempequationcounterPeakDATA}}
\hrulefill
%\vspace*{4pt}
\end{figure*}
\addtocounter{equation}{1}
\subsubsection{$r_i \leq N_r$}
UE $i$ has a dedicated-PoA link to an RS $r_i$ with rate differential $V_{r_i}$. In an RS-RS transmission, the UE transmits on both links to RS $r_i$ only. We thus have $a_i=0$ and (\ref{opt_func1}) equals $\min\bigg((V_{r_i}+q_i)^+, W_d\log_2\left(1+\frac{P^{(d)}_{i}}{E^{(d)}_{i}}\right)+$$W_a\log_2\left(1+\frac{P^{(a)}_{i}}{E^{(r)}_{i}}\right)\bigg)$. Conversely, when node $i$ transmits to the BS on the adaptive-PoA link (i.e. BS-RS transmission), we have $a_i=1$ and (\ref{opt_func1}) becomes $W_a\log_2\left(1+\frac{P^{(a)}_{i}}{E^{(b)}_{i}}\right) + \min\left((V_{r_i}+q_i)^+, W_d\log_2\left(1+\frac{P^{(d)}_{i}}{E^{(d)}_{i}}\right) \right)$. 
\subsubsection{$r_i =N_r+1$}
UE $i$ has a dedicated-PoA link to the BS. When the UE transmits on both links to the BS (i.e. BS-BS transmission), we have $a_i=1$ and (\ref{opt_func1}) equals $W_d\log_2\left(1+\frac{P^{(d)}_{i}}{E^{(d)}_{i}}\right) + W_a\log_2\left(1+\frac{P^{(a)}_{i}}{E^{(b)}_{i}}\right)$. Conversely, when the UE transmits to an RS on its adaptive-PoA link, we have $a_i=0$. This RS is one from which the UE receives the highest pilot signals. Thus, in a BS-RS transmission, equation (\ref{opt_func1}) becomes 
$\min\left((V_{r^A_i}+q_i)^+, W_a\log_2\left(1+\frac{P^{(a)}_{i}}{E^{(r)}_{i}}\right) \right)+$\\
$W_d\log_2\left(1+\frac{P^{(d)}_{i}}{E^{(d)}_{i}}\right)$, where $V_{r^A_i}$ is the rate differential at the RS.

We assume that the relays can broadcast the current values of the rate differentials on a downlink control channel, as in \cite{Ahmad2012}; the UEs then use the rate differential information to make adaptations. Each UE's decision can then be expressed as the following optimization problem: %in \eqref{opt_func2}.
\setcounter{equation}{8}
\begin{equation}
\begin{split}\label{opt_func2}
\underbrace{\max}_{a_i, P^{(d)}_{i}, P^{(a)}_{i}} f_i(a_i, P^{(d)}_{i}, P^{(a)}_{i}) \hspace*{2cm} \\
P^{(d)}_{i}+P^{(a)}_{i} \leq {P_{\max,i}} \hspace*{2cm}  \\
P^{(d)}_{i},P^{(a)}_{i} \geq 0 \hspace*{2.7cm} \\
a_i \in \{0,1\} \hspace*{2.1cm}\\
\end{split}
\end{equation}
The total transmit power of node $i$ is constrained by $P_{\max,i}$ over both of its links. {We next define expressions for the transmit power allocations that result in the peak (maximum) data rates in \eqref{peakData} for any transmission configuration: BS-RS (base station and relay), RS-RS (relay only) or BS-BS (base station only).} 
\begin{figure*}[!b]
\hrulefill
\end{figure*}
\newcounter{tempequationcounter5}
\begin{figure*}[!b]
\normalsize
\setcounter{tempequationcounter5}{\value{equation}}
\begin{equation*}
{P^{(d)}_i(k+1)=}\setcounter{equation}{11}
\end{equation*}
\addtocounter{equation}{0}
\begin{subequations}
\small{
\label{power_alloc}
\begin{empheq}[left={}\empheqlbrace]{align}
&\min\left(P_{\max,i}, \left(\frac{W_dP_{\max,i}-W_aE^{(d)}_i(k)+W_dE^{(r)}_i(k)}{W_a+W_d}\right)^{+}\right) \mbox{, if $a_i(k+1)=0$}\label{pa}\\
&\min\left(P_{\max,i}, \left(\frac{W_dP_{\max,i}-W_aE^{(d)}_i(k)+W_dE^{(b)}_i(k)}{W_a+W_d}\right)^{+}\right) \mbox{, if $a_i(k+1)=1$ and ${V_{r_i}(k)}\geq 0$}\label{pb} \\
&P^{(d)}_{i}(k) \mbox{, if $a_i(k+1)=1$ and ${-\widehat{V}\leq V_{r_i}(k)}< 0$} \label{pc}\\
&z\cdot P^{(d)}_{i}(k) \mbox{, if $a_i(k+1)=1$ and ${V_{r_i}(k)}< -\widehat{V}$} \label{pd}
\end{empheq}}
\end{subequations}
\hrulefill
\setcounter{equation}{\value{tempequationcounter5}}
%\vspace*{4pt}
\end{figure*}
\addtocounter{equation}{1}
\begin{proposition}
Given $E^{(d)}_{i}$, $E^{(b)}_{i}$ and $ E^{(r)}_{i}$, the peak data rate $\tau^{(rr)}_i$, $\tau^{(br)}_i$ or $\tau^{(bb)}_i$ for UE $i$ is achieved with a water-filling power allocation such that $P^{(d)}_{i} = \min\left(P_{\max,i},\left(\frac{W_d P_{\max,i}+W_dE^{(x)}_{i}-W_aE^{(d)}_{i}}{W_a+W_d}\right)^{+}\right)$ and $P^{(a)}_i=P_{\max,i}-P^{(d)}_i$ where $x \in \{b, r\}$ denotes the adaptive-PoA link to either the BS or RS.\label{a1}
\end{proposition}
The proof of this is in the Appendix. Note that the above may also represent the power allocation strategy of a UE when it does not have any information about the rate differential and maximizes its data rate on its two links.
\begin{corollary}\label{COR1}
The peak data rate $\tau^{(br)}_i > \tau^{(rr)}_i$ if $E^{(b)}_{i}< E^{(r)}_{i}$, and $\tau^{(br)}_i > \tau^{(bb)}_i$
if $E^{(r)}_{i}< E^{(b)}_{i}$.
\end{corollary}
The peak data rate achievable for UE $i: r_i\leq N_r$ through a BS-RS transmission is higher than that via RS-RS transmission if the effective interference to the BS is lower than that to the RS on the adaptive-PoA link. The same result applies for UE $i: r_i= N_r+1$ in relation to the peak data rate via BS-BS transmission.
\section{Network state-based Distributed Transmission (NDT)}
We assume that adaptations occur in time intervals denoted as $k \in \{1,2,..\}$. In interval $k$, the effective interferences $E^{(b)}_{i}(k)$, $E^{(d)}_{i}(k)$, $E^{(r)}_{i}(k)$ and the rate differential $V_r(k)$ at its associated relay are available to each UE $i$. 

{If the UEs were to apply a greedy algorithm to maximize \eqref{opt_func2}, each UE would make a locally optimal choice of transmit power levels on its two links and PoA selection. Due to the inter-dependence of interference and load on backhaul links, use of a greedy strategy by all UEs may not generally produce a globally optimal solution or converge}.

Instead we propose an alternative algorithm, which we call Network state-based Distributed Transmission (NDT), that approaches the objective in \eqref{opt_func2} and which outperforms the greedy algorithm. A positive-valued constant $\widehat{V}$, known to all UEs, is a rate differential threshold that represents some tolerable load at each relay. 
Under our strategy, a node selects its receiver on the adaptive-PoA link by comparing the available backhaul capacity against the peak data rates in \eqref{peakData} and the threshold $\widehat{V}$. This approach ensures that the achieved data rates do not become erratic as the UE autonomously adapt.

\subsection{UE $i: r_i\leq N_r$} We propose that UE $i$, which has a dedicated-PoA link to an RS, update its adaptive PoA as follows: 
\begin{displaymath}
{a_i}(k+1){=}
\end{displaymath}
\begin{subequations}
\label{aplus}
\small{
\begin{empheq}[left={}\empheqlbrace]{align}
  & 0, \parbox{15em}{ if ${V_{r_i}(k)}+q_i(k)\geq\tau^{(br)}_i(k)$ \\   and $\tau^{(rr)}_i(k)>\tau^{(br)}_i(k)$}\label{localadapta}\\[0.1cm]
  & 1, \parbox{15em}{ if $\tau^{(rr)}_i(k)\leq\tau^{(br)}_i(k)$ \\   or ${V_{r_i}(k)}+q_i(k)\leq\tau^{(br)}_i(k)-\widehat{V}$ }\label{localadaptb} \\[0.1cm]
  & a_i(k) \mbox{, otherwise.} \label{localadaptc}
\end{empheq}}
\end{subequations}
When ${V_{r_i}(k)}$ is sufficiently high, the choice of PoA between (\ref{localadapta}) and (\ref{localadaptb}) reduces to which transmission (RS-BS or RS-RS) can achieve a higher peak data rate. In \eqref{localadapta}, the UE choses an RS-RS transmission if the current backhaul capacity is large enough and if the achievable rate is more than that via a BS-RS transmission.

In (\ref{localadaptb}), the UE chooses a BS-RS transmission whenever $\tau^{(rr)}_i(k)\leq \tau^{(br)}_i(k)$ or if the current backhaul capacity is low. The load thresholds of each node are chosen such that the relay backhaul link remains a bottleneck link even after accounting for its own data rate. For all intermediate cases, in (\ref{localadaptc}), the PoA on the adaptive-PoA link remains unchanged. 

Corresponding to PoA choice in iteration $k+1$ in \eqref{aplus}, UE $i$ allocates transmit power to the dedicated-PoA link as shown in \eqref{power_alloc}. The first two cases \eqref{pa} and \eqref{pb} correspond to the waterfilling allocation for RS-RS and BS-RS transmissions, respectively, as indicated in Proposition~\ref{a1}. Given $z: 0<z<1$, \eqref{pd} is the power allocation on the dedicated-PoA (UE-to-RS) link when the relay backhaul link is overloaded. The node iteratively reduces its transmit power on its UE-to-RS link and re-allocates this power to its link to the BS. This continues until the backhaul load drops to an acceptable level such that $-\widehat{V}\leq V_{r_i}(k)<0$. As per \eqref{pc} then the UE can maintain its transmit power. Given its transmit power on the dedicated-PoA link, the remaining power is allocated by node $i$ to the adaptive-PoA link such that:
\addtocounter{equation}{1}
\begin{align}
\begin{split}\label{localadapt3}
P^{(a)}_{i}(k+1)&=P_{\max,i} - P^{(d)}_{i}(k+1) \\
\end{split}
\end{align}
\subsection{UE $i: r_i= N_r+1$} Next consider a UE $i$ that has a dedicated-PoA link to the BS. As before, its decision to select the RS on its adaptive-PoA link depends on whether the peak data rate via BS-RS transmission is higher than that via BS-BS transmission and on whether there is sufficient backhaul capacity (i.e. ${V_{r^A_i}(k)}+q_i(k)\geq -\widehat{V}$): 
\begin{displaymath}
{a_i}(k+1){=}
\end{displaymath}
\begin{subequations}
\label{aplus2}
\begin{empheq}[left={{}}\empheqlbrace]{align}
  & 0, \parbox{15em}{if ${V_{r^A_i}(k)}+q_i(k)> \widehat{V}$ and \\ $\tau^{(br)}_i(k)> \tau^{(bb)}_i(k)$}\label{localadapta2}\\[0.1cm]
  & 1,\parbox{15em}{if $\tau^{(bb)}_i(k)\leq\tau^{(br)}_i(k)$ or \\ ${V_{r^A_i}(k)}+q_i(k)\leq-\widehat{V}$}\label{localadaptb2} \\[0.1cm]
  & a_i(k) \mbox{, otherwise.} \label{localadaptc2}
\end{empheq}
\end{subequations}
If $a_i(k+1)=1$, then node $i$ performs BS-BS transmission with the corresponding transmit power allocation on the dedicated-PoA as \eqref{pb}. Conversely, if $a_i(k+1)=0$ then the node performs BS-RS transmission and allocates its power as per \eqref{pa}. As before, the remaining power is allocated to the adaptive-PoA link as per \eqref{localadapt3}. 

The key feature of the NDT algorithm is that it gives UEs the flexibility to switch between different transmission modes (i.e. RS-RS versus BS-RS transmission) based on both link-layer and backhaul load conditions. In contrast, in a fixed single connection (RS-RS and BS-BS) or dual connection (BS-RS) topology, the UEs cannot adapt to traffic load and link layer conditions. Moreover, when the backhaul links are over-loaded, conventional waterfilling on the links is not an optimal allocation. Under NDT the UEs re-allocate more transmit power (and data rate) on the direct link to the BS until the backhaul link is load-balanced. 

\section{Convergence} 
We analyze the conditions under which the distributed adaptations performed by the nodes under our NDT mechanism will converge to a solution of \eqref{opt_func2}. We define matrix ${\bf{A}}(k) = diag{({[a_1(k), a_2(k), \cdots, a_n(k)]})}$ which is an $n \times n$ diagonal matrix  where element $a_i(k)$, as defined earlier, represents whether node $i$ is connected to the BS or to an RS on its adaptive-PoA link. We let $\overline{\bf{A}}={\bf{I}}-{\bf{A}}$ where ${\bf{I}}$ is the identity matrix. For example, ${\bf{A}}(0)={\bf{I}}$ implies that all UEs initially transmit to the BS on their adaptive-PoA links. The vector representation of the effective interference in (\ref{effintf}) in iteration $k$ is then
\begin{equation} \label{effintfvec}
\begin{split}
{\bf{E}}_d(k) &= {\bf{D}}_d + {\bf{F}}{\bf{P}}_d(k) \mbox{, (dedicated-PoA links)},\\
{\bf{E}}_{a}(k) &= {\bf{A}}{\bf{D}}_{b} + \overline{\bf{A}}{\bf{D}}_{r} \mbox{, (adaptive-PoA links)}
\end{split}
\end{equation}
where ${\bf{E}}_{a}(k)(i)={\bf{D}}_b(i)$ when $a_i(k)=1$ and ${\bf{E}}_{a}(k)(i)={\bf{D}}_{r}(i)$ when $a_i(k)=0$. 
\subsection{High $\eta_r$ regime}
Note that under NDT when the capacity of the RS-to-BS backhaul links is sufficiently high, the choice of PoA of UEs depends only on the peak data rates achievable through the BS-RS, RS-RS or BS-BS transmissions. As per Corollary~\ref{COR1}, this in turn depends on the effective interferences on their adaptive-PoA links. 
\begin{corollary}
Given a large enough RS-to-BS link capacity $\eta_r$, $a_i=1$ when ${\bf{D}}_{b}(i)<{\bf{D}}_{r}(i)$ for each UE $i$ and conversely $a_i=0$ when ${\bf{D}}_{b}(i)\geq {\bf{D}}_{r}(i)$, regardless of the initial PoA assignments. \label{COR2}
\end{corollary}
\begin{proof}
Suppose that there is some $\eta_r$ which is more than the maximum aggregate data rate the UEs can ever forward through the relay. Under NDT, the choice of a PoA by each UE depends on the peak data rate it can achieve via waterfilling allocation on the two links. Recall that as per Corollary~\ref{COR1}, $\tau^{(rr)}_i< \tau^{br}_i$ when $E^{(b)}_i< E^{(r)}_i$ for UE $i:r_i\leq N_r$. Likewise, $\tau^{(bb)}_i< \tau^{br}_i$ when $E^{(b)}_i< E^{(r)}_i$ for UE $i:r_i=N_r+1$. If ${\bf{D}}_b(i)< {\bf{D}}_{r}(i)$ for UE $i$ then $E^{(b)}_i(k) < E^{(r)}_i(k)$ always. Hence the UE chooses the BS as the PoA under NDT on the adaptive-PoA link. Conversely, $E^{(b)}_i(k) \geq E^{(r)}_i(k)$ when ${\bf{D}}_b(i)\geq {\bf{D}}_{r}(i)$ and the UE will choose the RS. This applies regardless of which PoA the UE initially attaches to.
\end{proof}
Given a high enough backhaul capacity, we now derive a solution of \eqref{opt_func2} for each UE in the network under NDT. Corresponding to Corollary~\ref{COR2}, the choice of PoAs of all UEs on their adaptive-PoA links can be represented by ${\bf{A}}^{*}= diag {[a_1, a_2, \cdots a_n ]}$.
\begin{theorem}
The transmit powers converge to ${\bf{P}}_d^*=\frac{1}{W_a+W_d}\left[{\bf{I}}+\frac{W_a{\bf{F}}}{W_a+W_d}\right]^{-1}{\bigg[}{W_d\bf{P}}_{\max}-W_a{\bf{D}}_{d}+W_d{\bf{A}}^*{\bf{D}}_{b}+W_d\overline{\bf{A}}^*{\bf{D}}_{r}{\bigg]}$ given a large enough RS-to-BS link capacity $\eta_r$ and if ${\bf{0}}\leq {\bf{P}}_d^*\leq {\bf{P}}_{\max}$. \label{theoremx}
\end{theorem}
\begin{proof}
For some large $\eta_r$, the rate differential at each relay is larger than the thresholds set in (\ref{localadapta}) and in (\ref{localadaptb}). Hence, the choice of PoA on the adaptive-PoA link depends only on the peak data rates $\tau^{(rr)}_i$ and $\tau^{(br)}_i$ if $r_i \leq N_r$. Conversely, if $r_i=N_r+1$ then the choice is between $\tau^{(bb)}_i$ and $\tau^{(br)}_i$. Hence, regardless of the transmit powers, as per Corollary~\ref{COR2}, let ${\bf{A}}^*$ represent the choice of PoAs given the normalized noise powers on the channels for adaptive-PoA links:  the nodes allocate transmit power according to either \eqref{pa} or \eqref{pb}. We represent the power updates for the dedicated-PoA links in matrix notation as 
\begin{displaymath}
{\bf{P}}_d(k+1)=\frac{\left[W_d{\bf{P}}_{\max}-W_a{\bf{E}}_d(k)+W_d{\bf{E}}_{a}(k)\right]}{W_a+W_d}
\end{displaymath}
\begin{displaymath}
\hspace*{0.9cm}=\frac{1}{W_a+W_d}{[}{W_d\bf{P}}_{\max}-W_a{\bf{D}}_d - W_a{\bf{F}}{\bf{P}}_d(k)+\\
\end{displaymath}
\begin{displaymath}
\hspace{0.12in}W_d{\bf{A}}^*{\bf{D}}_{b} + W_d \overline{{\bf{A}}^*}{\bf{D}}_{r}{]}\\
\end{displaymath}
\begin{equation}
\begin{split} \label{update_ue_rs}
&=\frac{{[}{W_d\bf{P}}_{\max}-W_a{\bf{D}}_{d}+W_d{\bf{A}}^*{\bf{D}}_{b}+W_d\overline{\bf{A}}^*{\bf{D}}_{r}{]}}{{W_a+W_d}}\\
&\hspace{0.15cm} - \frac{W_a}{W_a+W_d}{\bf{F}}{\bf{P}}_d(k) \\
&={\bf{N}}-{\bf{M}}{\bf{P}}_d(k)
\end{split}
\end{equation}
where ${\bf{N}}= \frac{1}{W_a+W_d}{[}W_d{\bf{P}}_{\max}-W_a{\bf{D}}_{d}+W_d{\bf{A}}^*{\bf{D}}_{b}+W_d\overline{\bf{A}}^*{\bf{D}}_{r}{]}$ and ${\bf{M}}= \frac{W_a}{W_a+W_d}{\bf{F}}$. 
The above evolves to:
\begin{equation}\label{powerseries}
\begin{split}
{\bf{P}}_d(k+1)&={\bf{N}}-{\bf{M}}\left({\bf{N}} - {\bf{M}}\left({\bf{N}}-{\bf{M}}\left(\cdots {\bf{P}}(0) \right) \right)\right) \\
&={[}{\bf{I}} - {\bf{M}}+{\bf{M}}^2-\cdots{]}{\bf{N}}+{\bf{M}}^{k}{\bf{P}}_d(0)\\
\lim\limits_{k \rightarrow \infty} {\bf{P}}_d(k+1)&={\bf{P}}^*_d=\left[{\bf{I+M}}\right]^{-1}{\bf{N}}\\
&=\frac{1}{W_a+W_d}\left[{\bf{I}}+\frac{W_a{\bf{F}}}{W_a+W_d}\right]^{-1}\\
{\bigg[}W_d{\bf{P}}_{\max} - &W_a{\bf{D}}_{d}+W_d{\bf{A}}^*{\bf{D}}_{b}+W_d\overline{\bf{A}}^*{\bf{D}}_{r}{\bigg]}
\end{split}
\end{equation}
if the spectral radius (maximum absolute eigenvalue) of $\frac{W_a}{W_a+W_d}{\bf{F}}$ is less than one and where ${\bf{P}}_d(0)$ is the initial transmit power vector of the dedicated-PoA links. Likewise, the corresponding transmit powers on the adaptive-PoA channels are simply the difference between ${\bf{P}}_{\max}$ and  (\ref{update_ue_rs}). They evolve to 
$\lim\limits_{k \rightarrow \infty} {\bf{P}}_a(k)={\bf{P}}_a^*={\bf{P}}_{\max}- {\bf{P}}^*_d$.
\end{proof}
\subsection{Limited $\eta_r$ regime}
Next we consider the general case where the backhaul capacity is limited. 
\begin{theorem}
Given a large enough tolerable $\widehat{V}$, NDT adaptations will converge if the spectral radius of ${\bf{F}}$ is less than $\frac{W_a+W_d}{W_a}$. \label{finiteConvf}
\end{theorem}
\begin{proof}
According to Corollary~\ref{COR2} the inequality between the peak data rates $\tau^{(rr)}_i$, $\tau^{(br)}_i$ or $\tau^{(bb)}_i$ always holds given the inequality relationship between ${\bf{D}}_b(i)$ and ${\bf{D}}_r(i)$ for each UE $i$. Thus, given a high enough $\widehat{V}$, the nodes will not cycle back and forth their links between different PoAs as per \eqref{aplus} and \eqref{aplus2}. The power allocation of UE $i: r_i= N_r+1$ will only follow  the waterfilling allocation either in \eqref{pb} or \eqref{pa}. Likewise, if UE $i: r_i\leq N_r$, then this node will either perform waterfilling allocation as per \eqref{pa}\eqref{pb} or maintain its transmit power from the preceding iteration as per \eqref{pc}. Thus, the evolution of transmit power for the system is equivalent to that in \eqref{powerseries} except that some nodes may not update their transmit power every iteration (i.e. their transmit power level is based on an update from an earlier iteration). The power updates can therefore be described as an \emph{asynchronous iterative} system \cite{AsyncIterSysAhmad2013}\cite{waterfilling_Shum07}. It is known that such a system converges if the iterative matrix $\bf{M}$ in \eqref{powerseries} has a spectral radius less than one \cite{waterfilling_Shum07} (corresponding to a spectral radius of less than $\frac{W_a+W_d}{W_a}$ for ${\bf{F}}$). 
\end{proof}
\newcounter{optimalN}
\begin{figure*}[!b]
\hrulefill
\setcounter{optimalN}{\value{equation}}
\setcounter{equation}{17}
\begin{equation}{\small
\begin{split}
\eta^{(A)}_N &= \underset{P^{(d)}_i P^{(a)}_i a_i}{\max}{\hspace{0.2cm}} {W_d\sum^{n=2}_{i=1}\log_2(1+P^{(d)}_i\frac{\sigma^2_i}{n_d})+W_a\left(\sum^{n=2}_{i=1} \eta\left( P^{(a)}_i\frac{g_{i,r}+g_{i,b}}{n_a}\right)\right)} \hspace*{0cm}  \\
\eta^{(B)}_N &= \underset{P^{(d)}_i P^{(a)}_i a_i}{\max}\hspace{0.2cm} W_d\eta\left(\frac{g^{(1)}_{1,1}P^{d}_1}{n_d}\right) + W_d\eta\left(\frac{g^{(1)}_{2,2}P^{(d)}_2}{n_d+g^{(1)}_{1,2}P^{(d)}_1}\right) + W_a\left(\sum^{n=2}_{i=1} (1-a_i)\eta\left(\frac{g_{i,r} P^{(a)}_i}{n_a}\right)+\sum^{n=2}_{i=1} a_i\eta\left(\frac{g_{i,b} P^{(a)}_i}{n_a}\right)\right) \hspace*{0.4cm} \\
\eta^{(C)}_N &= \underset{P^{(d)}_i P^{(a)}_i a_i}{\max}{\hspace{0.1cm}}W_d \sum^2_{i=1} \eta\left(\frac{g^{(1)}_{i,i}P^{d}_i}{n_d+ \sum^2_{j=1;j\neq i} g^{(1)}_{j,i}P^{d}_j}\right) +W_a\left(\sum^{n=2}_{i=1} (1-a_i)\eta\left(\frac{g_{i,r}P^{(a)}_i}{n_o}\right)+\sum^{n=2}_{i=1} a_i\eta\left(\frac{g_{i,b} P^{(a)}_i}{n_o}\right)\right)\hspace*{1.2cm} \\
P^{(d)}_i+P^{(a)}_i\leq P_{max,i}, \hspace*{-8.8cm}\\
P^{(d)}_i,P^{(a)}_i\geq 0,  \hspace*{-8cm}\\
a_i \in \{0,1\}, \hspace*{-8.5cm}\\
i\in\{1,2\} \hspace*{-8.6cm}
\end{split}}\label{Optimality}
\end{equation}
\setcounter{optimalN}{\value{equation}}
\hrulefill 
\end{figure*}

\begin{figure}[!t]
\subfigure[] {\includegraphics[width=0.5\textwidth]{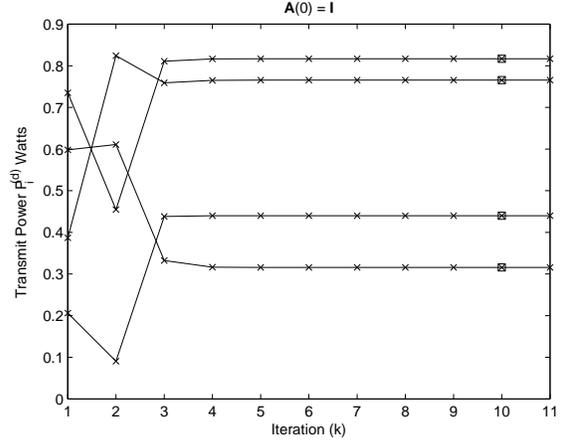}}
\subfigure[] {\includegraphics[width=0.5\textwidth]{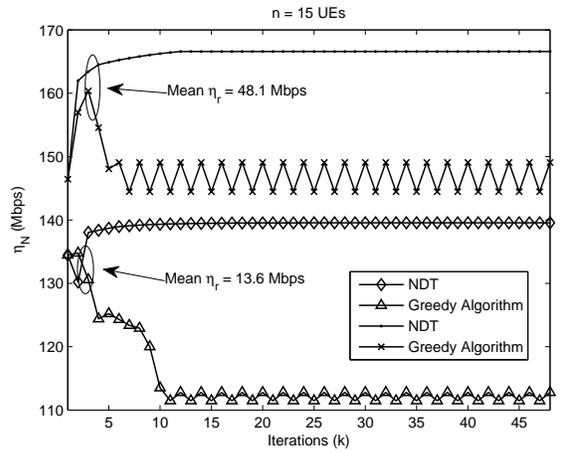}}
\caption{(a) We observe for the network shown in Fig. \ref{fig:furt1}, the convergence of the transmit power vector ${\bf{P}}_d$ to the fixed point (superimposed on $k=10$) predicted in (\ref{powerseries}) regardless of the initial power and PoA assignments. {(b) For another sample network, we plot the evolution of network capacity and see that NDT converges within a few iterations whereas a greedy algorithm results in unstable performance.}}
\label{fig:furt2}
\end{figure}

\subsection{Practical Considerations}
{In wireless systems, feedback delay and estimation error may render CSIT imperfect \cite{Medard2000}. The adaptations will thus be based on inaccurate CSIT that may render our scheme less effective in improving network capacity. Given the estimated gain $\widehat{g^{(f)}_{i,r}}$ for UE $i$, estimation error will project the effective interference vectors in \eqref{effintfvec} to:
\begin{equation}
\begin{split}
\widehat{{\bf{E}}_d}(k) &= \widehat{{\bf{D}}_d} + \widehat{{\bf{F}}}{\bf{P}}_d(k),\\
\widehat{{\bf{E}}_{a}}(k) &= {\bf{A}}\widehat{{\bf{D}}_{b}} + \overline{\bf{A}}\widehat{{\bf{D}}_{r}}
\end{split}
\end{equation}
Thus, corresponding to the imperfect channel knowledge, if the spectral radius of $\widehat{{\bf{F}}}$ is less than $\frac{W_a+W_d}{W_a}$ then the results in Corollary~\ref{COR2}, Theorem~\ref{theoremx} and \ref{finiteConvf} still hold. Therefore, under NDT, adaptation would converge with the system vectors adjusted for estimation error. Moreover, UEs could switch between different transmission modes every iteration using the multipoint signaling feedback mechanism provided in the LTE architecture \cite{3gppComP12}. Hence, implementing our proposed scheme will not require any additional overhead.}

\section{Performance Optimality}
We now discuss some relevant notions of optimality based on either interference cancellation or cooperative communication with centralized optimization, so as to provide achievable rate comparison with NDT. As in \cite{sahin2011_icobr, yeTian2012_icobr, Sahin2011_icobr2, Razaghi2013_ICOBR}, we also consider a network with $n=2$ UEs, and $N_r=1$ relay. The first UE has a dedicated-PoA link to the BS and the other has a dedicated-PoA link to the RS. The two dedicated-PoA links operate on the same channel. The backhaul capacity $\eta_r$ is assumed high enough such that it does not become rate-limiting. 

\subsection{Cooperative MIMO (CO-MIMO)}
As shown in \cite{Gesbert10}, a multi-user cooperative network can be viewed as a virtual MIMO system. In our system, the two UEs collaborate as the transmitter set and the PoAs (BS and the RS) as the joint receivers in a $2\times 2$ MIMO. Given the channel matrix  
${\bf{H}}=\begin{bmatrix}
       h^{(1)}_{1,1} & h^{(1)}_{1,2}        \\[0.3em]
       h^{(1)}_{2,1} & h^{(1)}_{2,2}         \\[0.3em]
     \end{bmatrix}$, {let its $i^{th}$ singular value obtained via Singular Value Decomposition (SVD) be $\sigma_i$.} The maximum sum rate $\eta^{(A)}_N$ is expressed in \eqref{Optimality}, where the left-most term denotes the {equivalent MIMO channel capacity using SVD \cite{goldsmith05} of the dedicated-PoA links and the other summation term indicates the rate achieved through maximal ratio diversity combining on the adaptive-PoA links.} 

\subsection{Asymmetric Interference Cancellation (AIC)}
In AIC, the BS can perform interference cancellation whereas the RS does not have any such capability. The maximum sum rate is bounded by $\eta^{(B)}_N$ in \eqref{Optimality} where the left-most term denotes the achievable rate for the dedicated-PoA UE-to-BS link when all interference received from the UE-to-RS dedicated link can be eliminated \cite{EtkinTse2008_GIC}. 

\subsection{Cross-Layer Optimization (C-OPT)}
In C-OPT, the optimal choice of PoAs and power allocations is determined to maximize the sum rate without any cooperation or interference cancellation and is expressed as $\eta^{(C)}_N$ in \eqref{Optimality}. \\

\begin{figure}[!t]
\subfigure[]{\includegraphics[width=0.48\textwidth]{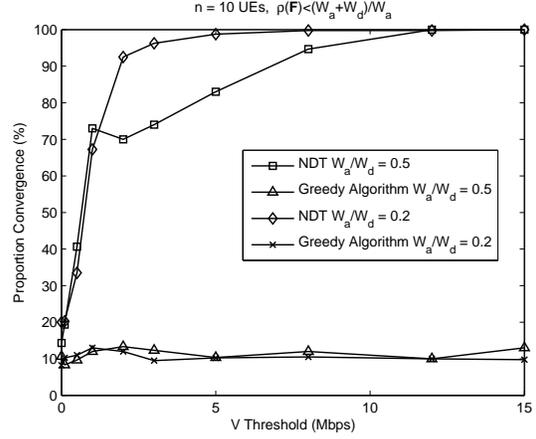}}
\subfigure[]{\includegraphics[width=0.48\textwidth]{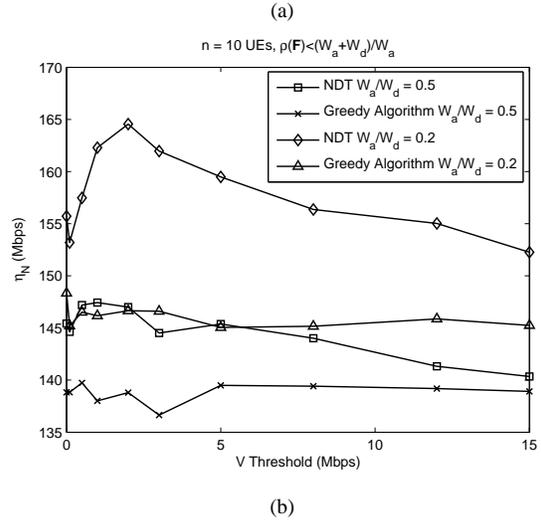}}
\caption{(a) With a higher $\widehat{V}$, the system always converges if the spectral radius constraint on $\bf{F}$ holds. The corresponding network capacity is shown in (b).}\label{res3}
\end{figure}
In the next section, through numerical results we compare the performance under CO-MIMO, AIC and C-OPT against NDT.

\section{Simulation Results}
In this section, for a multihop cellular network we use Matlab-based simulation to determine the aggregate end-to-end data rate $\eta_N$ over a range of control parameters. We assume that the cross-link power gains {of the UEs} ($g^{(f)}_{i, r_i}, g^{(q)}_{i, b}$, etc.) are of the form $\kappa^{(f)}_{j, i}\cdot d_{j, i}^{-\alpha}$ where $\kappa^{(f)}_{j, i}$ is an exponentially distributed random variable with unit variance (due to Rayleigh fading) on channel $f$ and $\alpha$ is the path loss exponent. The fading gains are assumed independent and identically distributed for all links. There are $50$ power control iterations in each simulation trial. The transmit powers of pilot signals from the BS and RS on the downlink control channel are assumed equal. The UEs are randomly and uniformly located in cluster regions of radius $R_L$ around each of the $N_r$ relays with $P_{\max,i}=1.0$ Watts. 
\begin{figure}[!t]
\subfigure[]{\includegraphics[width=0.5\textwidth]{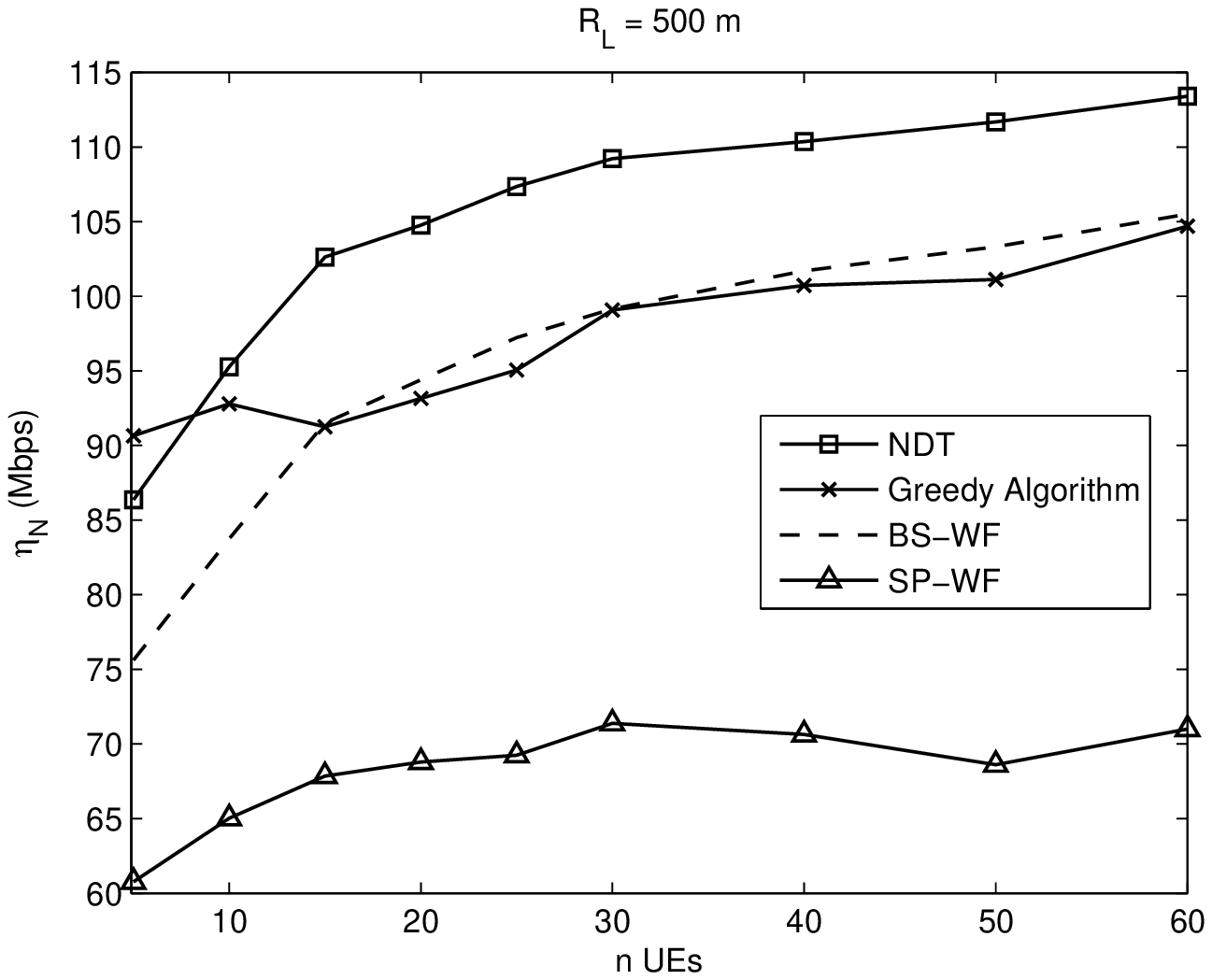}}
\subfigure[]{\includegraphics[width=0.5\textwidth]{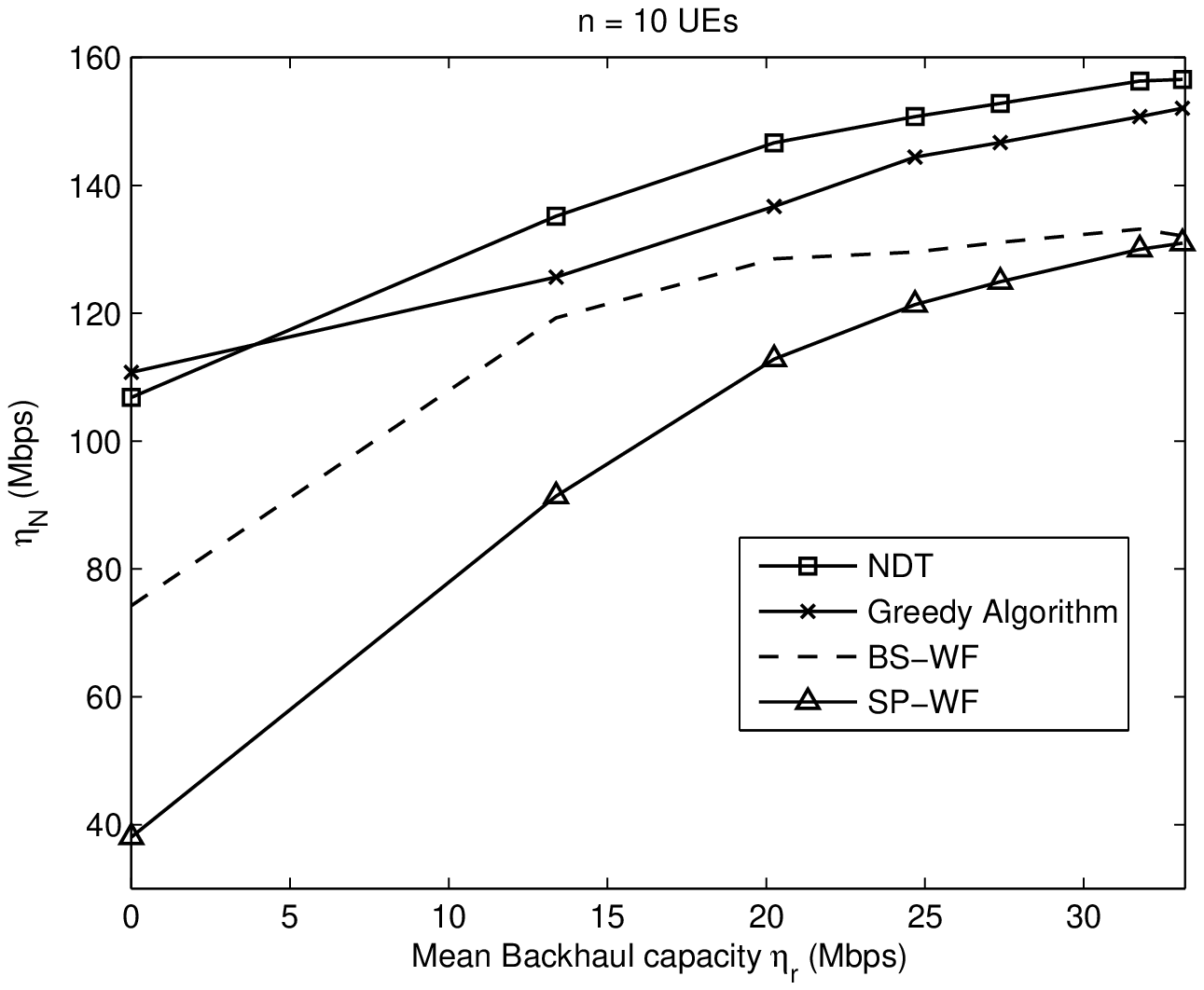}}
\caption{{Over the range of (a) $n$ or (b) $\eta_r$, BS-WF or SP-WF schemes have good performance over certain regions only. In contrast, NDT enables nodes to benefit from the best of both schemes and also adapt to the load on the backhaul links.}}\label{res1}
\end{figure}

Unless otherwise stated, we assume the following parameters: for NDT $z=0.8$, $\widehat{V}=5$ Mbps,  $R_L=250$ m, {noise power spectral density is $-200$ dBW/Hz (proportional to channel bandwidth \cite{verdu_2002}), $W_s = 10$ MHz and $\alpha=4.0$.}
The backhaul links operate on reserved bands from an operator's pool that are non-overlapping with channels used by the UEs. {We use the backhaul channel model based on the 3GPPTR 36.814 specifications \cite{3gpp2010}. The line-of-sight (LOS) and non-line-of-sight (NLOS) path losses of each wireless relay backhaul link are based on the urban macro model based on carrier frequency $2$ GHz, relay height $5$ m with all remaining parameters (e.g. transmit power, bandwidth or lognormal shadowing variance etc.) based on standard values specified in \cite[p.~72, p.~94]{3gpp2010}. The backhaul capacity $\eta_r$ is then determined based on Shannon capacity.} When $n$ is large, there is less bandwidth allocated to each UE. The division of the allocated bandwidth between the dedicated-PoA and adaptive-PoA links for each UE is set at $\frac{W_a}{W_d}=0.5$, unless stated otherwise. There are $N_r=4$ relays located randomly in a $2$ km by $2$ km square region around the BS with non-overlapping cluster regions.
\begin{figure}[!t]
\subfigure[]{\includegraphics[width=0.5\textwidth]{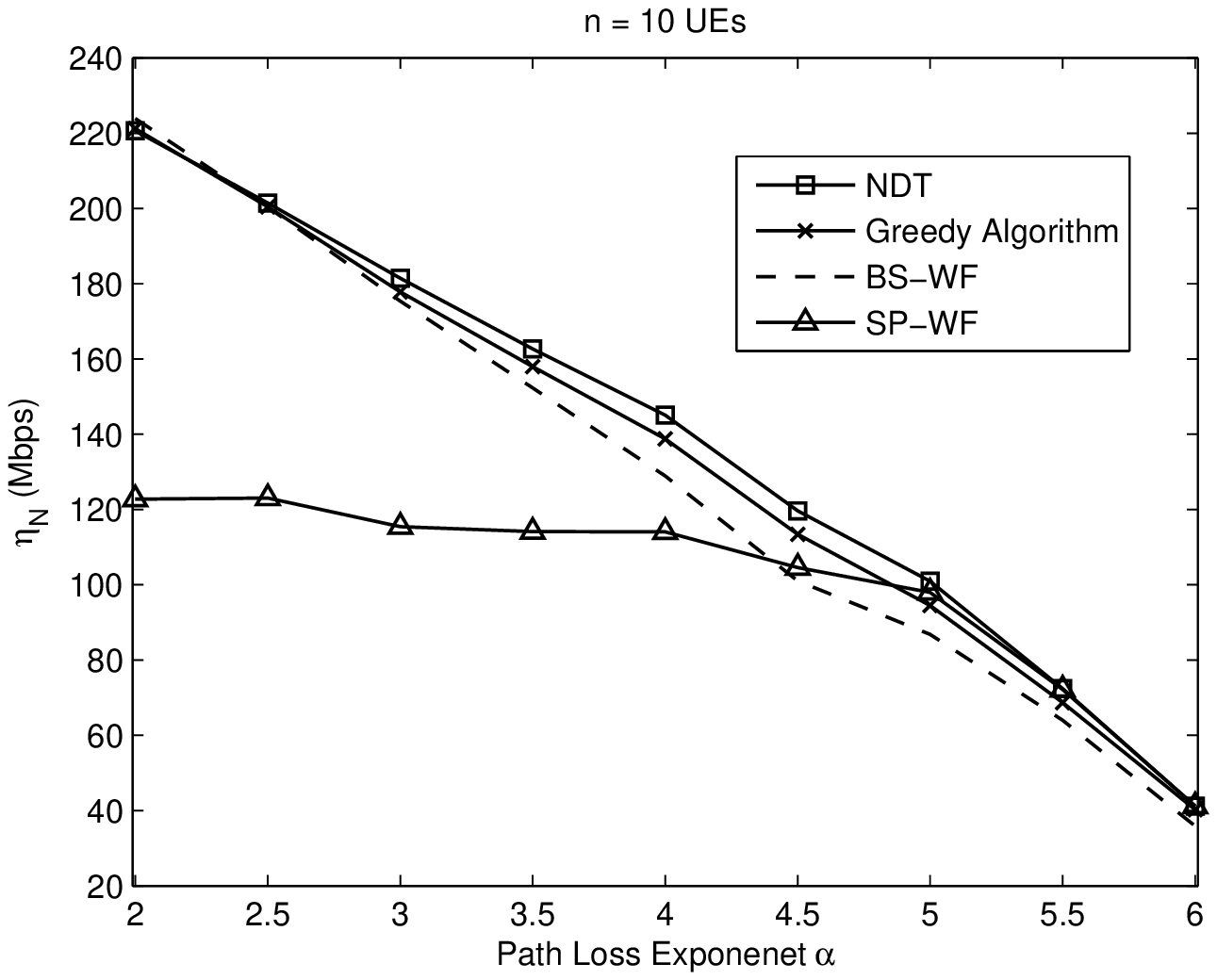}}
\subfigure[]{\includegraphics[width=0.5\textwidth]{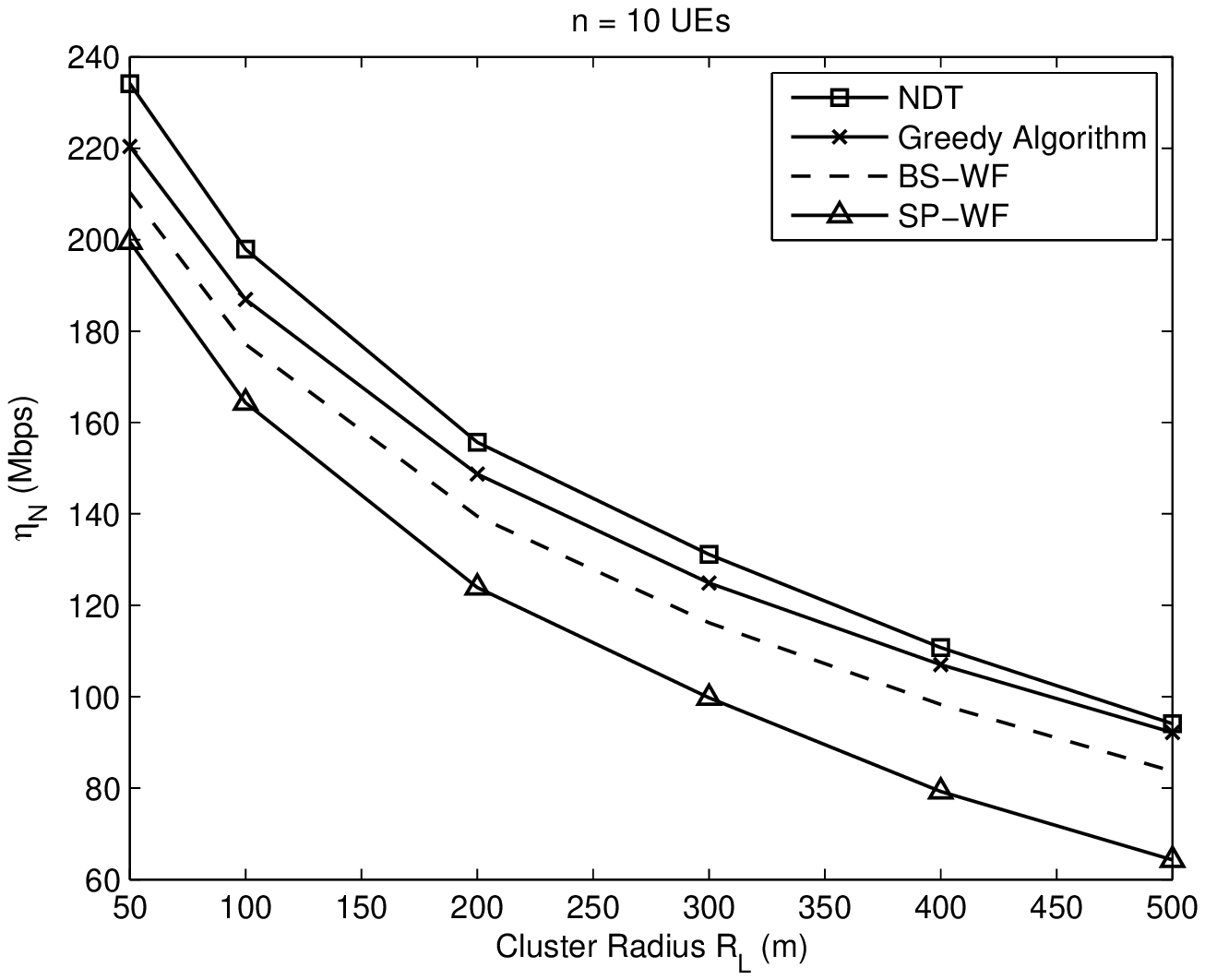}}
\caption{{(a) With increasing $\alpha$ or (b) when the cluster radius $R_L$ is larger so that the UEs are dispersed over a wider region, the path loss is increased, thereby leading to a performance decline for all schemes. NDT enables the nodes to better adapt over the range of these conditions.}}\label{res2}
\end{figure}

\begin{figure}[!t]
\subfigure[]{\includegraphics[width=0.5\textwidth]{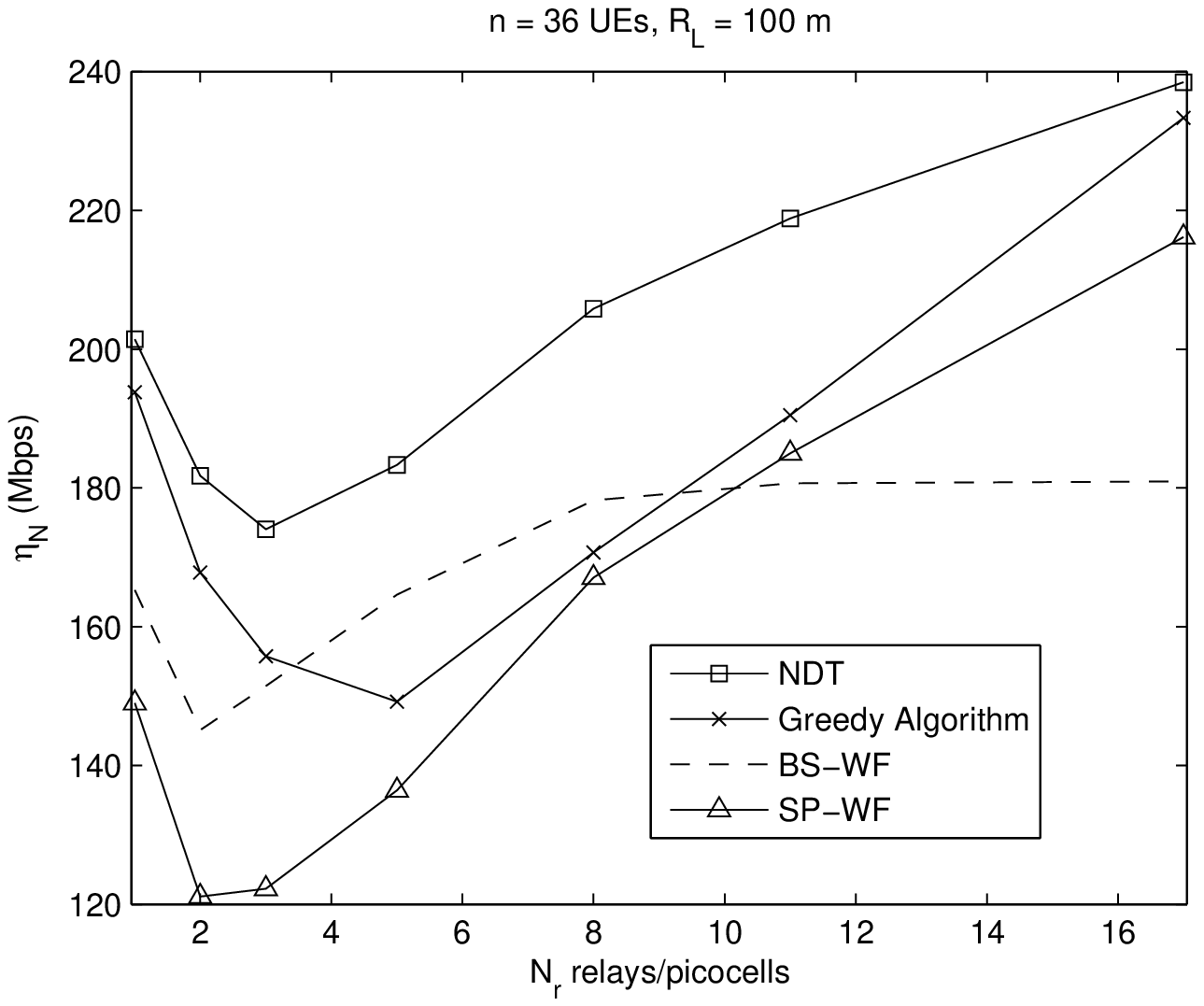}}
\subfigure[]{\includegraphics[width=0.5\textwidth]{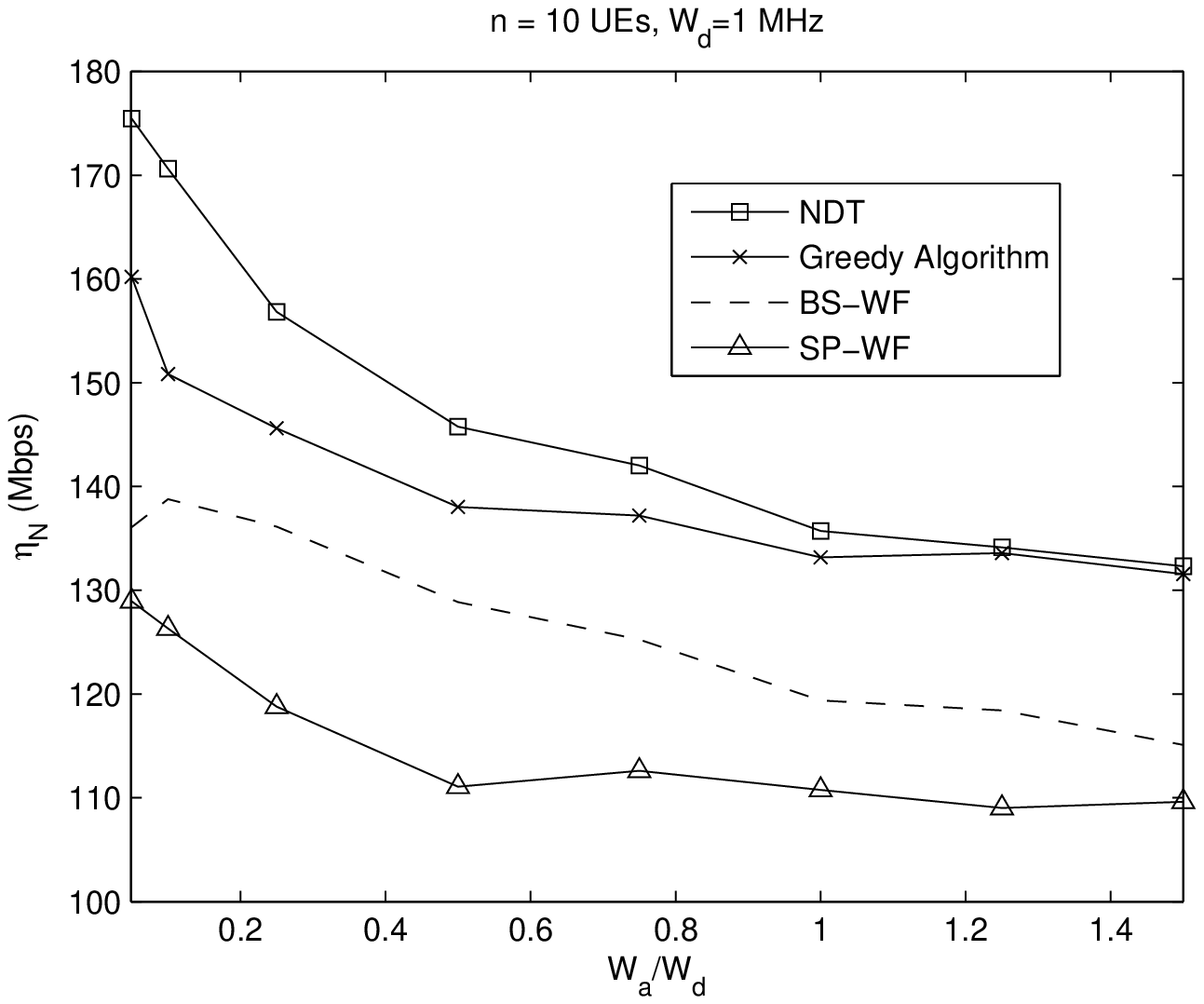}}
\caption{{(a) With increasing number of relays that have UEs clustered close to them, the multihop network approximates the case of small picocells overlaid within a macro-cell. (b) Given a fixed system bandwidth, the performance of all schemes declines with increasing $W_a/W_d$ as there is less frequency re-use.} }
\label{res4}
\end{figure}

\subsection{Convergence}
As an example of convergence, we first consider the network in Fig. \ref{fig:furt1} with $n=3$ UEs located in the positions shown. We plot the evolution of transmit powers given that initially all UEs transmit on their adaptive-PoA link to their relays with a randomly chosen initial vector ${\bf{P}}_d(0)$. In Fig. \ref{fig:furt2}(a), we observe that the transmit powers converge to the vector predicted in (\ref{powerseries}) where ${\bf{A}}^*=diag [0,0,0,1]$. In Fig. \ref{fig:furt2}(b), we plot time evolution of the network-wide performance $\eta_N$ for two sample networks. We compare NDT against approaching \eqref{opt_func2} through a {greedy algorithm}, which represents an immediate solution for $a_i(k+1), P^{(d)}_i(k+1)$ and $P^{(a)}_i(k+1)$ for UE $i$. The greedy approach, in contrast to NDT, results in an erratic evolution of the transmit powers, which correspondingly causes an unstable and reduced network capacity over time. 

In Fig. \ref{res3}, we plot the effects of the rate differential threshold $\widehat{V}$ on network capacity and the corresponding likelihood of convergence for those trials where the spectral radius $\rho(\bf{F})$ is less than $\frac{W_a+W_d}{W_a}$ (recall Theorem~\ref{finiteConvf}). With a higher $\widehat{V}$, NDT will converge $100\%$ of the time. In contrast, convergence under a greedy algorithm occurs much less frequently and results in lower performance. 

Note that in the intermediate case, when the backhaul capacity is limited and $\widehat{V}$ is small, then convergence cannot always be guaranteed as shown in Fig. \ref{res3}a. For instance, to alleviate load on the backhaul links, UEs can send more data rate directly to the BS. Thus, if $\widehat{V}$ is set too small the data rates sent via relays  decrease too fast and the backhaul links may become under-loaded. In turn this will induce the UEs to switch transmission mode and send all their data via the relays. Thus, with a small $\widehat{V}$, there may be oscillation between over-loading the backhaul links or under-utilizing them as UEs make local adaptations with instantaneous knowledge of the  system. Nonetheless, we observe in Fig. \ref{res3}b that even in these cases the likelihood of non-convergence is still relatively small.

\subsection{Network Capacity}
{We compare NDT against two schemes: \\
(1) \emph{Single-PoA Waterfilling (SP-WF)}: As in \cite{waterfilling_Shum07} each node has a single PoA; if attached to an RS, it allocates transmit power as per \eqref{pa}\eqref{localadapt3} whereas if attached to the BS it uses \eqref{pb}\eqref{localadapt3}. \\
(2) \emph{BS-PoA  Waterfilling (BS-WF)}: UEs at cell edge transmit to RS and BS both as in \cite{LiangPoor2007_resAlloc_maxmin}, whereas the remaining UEs only transmit to the BS \cite{waterfilling_Shum07}. Either type of UE allocates its transmit power as per \eqref{pb} and \eqref{localadapt3}. }

Fig. \ref{res1} shows that, in contrast to BS-WF or SP-WF, NDT offers the best of both approaches over the entire range of the parameters $n$ and $\eta_r$. Under NDT, when the load is lighter (either $\eta_r$ is larger or $n$ is smaller), the UEs only transmit to a single PoA whereas, conversely, they switch to transmitting to the BS on their adaptive-PoA links when the load is heavier. Another observation is that, when the load is large, under NDT the network performance is better than for BS-WF. This is because under NDT, a UE uses power control in \eqref{pd} and iteratively re-allocates more power on the direct link to the BS, instead of simply waterfilling over its channels. {Note that in Fig. \ref{res1}b, we obtain the mean $\eta_r$ by varying the bandwidth allocated to each relay backhaul link in the operator pool up to 30 MHz.}

In Fig. \ref{res2}, we plot the performance by varying the path loss exponent $\alpha$ and cluster radius $R_L$. We observe that, while performance for all schemes declines with an increase in either parameter, NDT still offers the best of any of the schemes. In Fig. \ref{res2}(b), with increasing $R_L$, the inter-UE distance grows and the nodes become dispersed over a wider region, and thus more distant from their receivers. While the performance for NDT declines with increasing $R_L$ (due to greater path loss), it is still better than that of other schemes. 
\begin{figure}[!t]
{\includegraphics[width=0.5\textwidth]{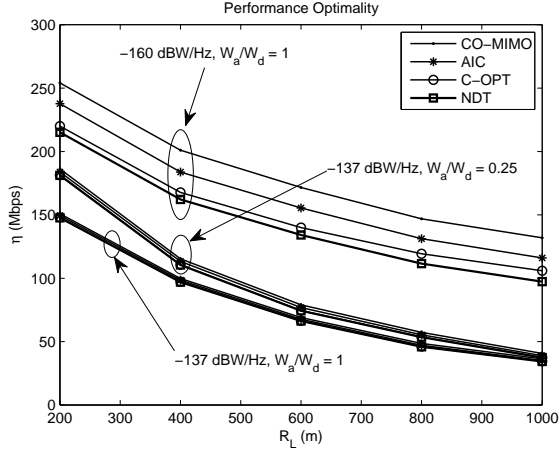}}
\caption{{In noisier systems or when the channel bandwidth of the adaptive-PoA links is larger, the performance gap between NDT and CO-MIMO, AIC and C-OPT diminishes.}}
\label{res5}
\end{figure}

In Fig. \ref{res4}(a) we consider the case with a large number of relays. Since the RS-to-BS links operate on non-UE channels, the 2-hop network would resemble the situation when picocells coexist within a macro-cellular system. Moreover, the backhaul capacity of the picocells is constrained \cite{Amitava2012,WeiYu2013} and is variable. Hence, in this scenario the relays of the 2-hop network could resemble picocell base stations. %Thus, for a UE connected to a picocell/femtocell base station, the end-to-end data rate would approximate that which it achieves if it was connected to a relay in our 2-hop cellular system.
The UEs are evenly spread between the PoAs and where the maximum RS-BS separation is $R=1000$ m and $R_L=100$ m. An implication from the result is that when there are more small cells and the UEs can transmit to both the macro-cellular BS and the picocell BS then there are performance gains over when they are constrained to transmit only to picocell base stations. In Fig. \ref{res4}(b), we consider the impact of the channel bandwidth size of the adaptive-PoA links relative to the dedicated-PoA links. When $W_a$ is large the network capacity of all schemes becomes better. 
 
Finally, in Fig. \ref{res5}b, we compare network capacity under NDT against the mechanisms discussed in Section VII, for a 2-user 1-relay system. Each UE is located exactly $R_L$ m from its dedicated PoA perpendicular to the line representing BS-RS separation (of 2000 m). The CO-MIMO and AIC schemes essentially decode interference apart from the desired signal, which enables them much to achieve much higher rates than both NDT and cross-layer optimization. However, we observe that if the system becomes noise-limited or when the relative bandwidth of the adaptive-PoA links is larger, NDT approaches the optimal achievable performance. This is because more noise or a larger $W_a$  decreases the relative advantage of interference cancellation capability of AIC and CO-MIMO.

\section{Conclusion}
We have proposed, for a 2-hop cellular uplink network, a link adaptive scheme where nodes have the flexibility to transmit to multiple PoAs. Under our scheme, nodes either transmit to a single PoA or split their data streams between the base station and a relay, using CSIT and backhaul load condition. We demonstrate a significant improvement in the aggregate end-to-end data rate in the network due to the proposed adaptation scheme. Future work could consider coordination between the UEs to form cooperative MIMO links and the impact of an adaptive bandwidth allocation to the backhaul links. 

\section*{Acknowledgement}
We would like to thank Dr Ravi Tandon for constructive feedback. This work is partially supported by the Science Foundation Ireland under Grant No. 10/IN.1/I3007.
\appendix
\subsubsection*{Proof of Proposition~\ref{a1}}
Consider the peak RS-RS transmission data rate $\tau^{(rr)}_i$ for UE $i$:
\begin{eqnarray*} 
\begin{split}
\underbrace{\max}_{P^{(a)}_{i}, P^{(d)}_{i}} W_a\log_2\left(1+\frac{P^{(a)}_{i}}{E^{(r)}_{i}}\right)+W_d\log_2\left(1+\frac{P^{(d)}_{i}}{E^{(d)}_{i}}\right)\\
P^{(a)}_{i}+P^{(d)}_{i}\leq P_{\max,i}\hspace*{1in}\\
P^{(a)}_{i}, P^{(d)}_{i} \geq 0 \hspace*{1.3in}\label{localOPT1}
\end{split}
\end{eqnarray*}
Using Lagrangian multipliers and KKT conditions we obtain:
\begin{eqnarray*} \small
\begin{split}
&\underbrace{\max}_{P^{(d)}_{i}, P^{(a)}_{i}} W_d\log_2\left(1+\frac{P^{(d)}_{i}}{E^{(d)}_{i}}\right) + W_a\log_2\left(1+\frac{P^{(a)}_{i}}{E^{(a)}_{i}}\right)\\
&-\lambda_1(P^{(d)}_{i}+P^{(a)}_{i}- P_{\max,i})-\lambda_2 P^{(d)}_{i} -\lambda_3 P^{(d)}_{i}\\
&\frac{W_d}{\ln(2)}.\frac{1}{P^{(d)}_{i}+E^{(d)}_{i}}-\lambda_1-\lambda_2 =0\\
&\frac{W_a}{\ln(2)}.\frac{1}{P^{(a)}_{i}+E^{(a)}_{i}}-\lambda_1-\lambda_3 =0\\
&\lambda_1(P^{(d)}_{i}+P^{(a)}_{i}- P_{\max,i}) =0\\
&\lambda_2 P^{(d)}_{i}=0\\
&\lambda_3 P^{(a)}_{i}=0
\end{split}
\end{eqnarray*}
Case 1: With $\lambda_2=\lambda_3=0, \lambda_1 \neq 0$, we obtain:
\begin{eqnarray*} 
\begin{split}
&\frac{W_d}{\ln(2)}.\frac{1}{P^{(d)}_{i}+E^{(d)}_{i}}-\frac{W_a}{\ln(2)}.\frac{1}{P_{\max,i}-P^{(d)}_{i}+E^{(a)}_{i}}=0\\
& P^{(d)}_{i} = \frac{W_d P_{\max,i}+ W_d E^{(a)}_{i}- W_a E^{(d)}_{i}}{W_d+W_a}\\
& P^{(a)}_{i} ={P_{\max,i}-P^{(d)}_i}\\
& =\frac{W_d P_{\max,i}-W_d E^{(a)}_{i}+ W_a E^{(d)}_{i})}{W_d+W_a}
\end{split}
\end{eqnarray*}
Case 2: With $\lambda_3 \neq 0, \lambda_2 = 0$ (i.e. $P^{(d)}_{i}=0$) we obtain:
\begin{eqnarray*} 
\begin{split}
\lambda_1=\frac{W_d}{\ln(2)}.\frac{1}{P^{(d)}_{i}+E^{(d)}_{i}}\\
P^{(d)}_{i}=0\\
P^{(a)}_{i}+0=P_{\max,i}
\end{split}
\end{eqnarray*}
We thus have $P^{(d)}_{i} = \min\left(P_{\max,i},\frac{\left(W_dP_{\max,i}-W_a E^{(d)}_{i}+ W_d E^{(r)}_{i}\right)^+}{W_a+W_d}\right)$.
The peak BS-RS or BS-BS transmission data rates $\tau^{(br)}_i$ and $\tau^{(bb)}_i$ have the same allocations as their  maximization is identical.

\bibliography{References}

\end{document}